\documentclass[11pt, a4paper]{article}

\usepackage{graphicx}
\usepackage{here}
\usepackage{amsmath}
\usepackage{amsfonts}
\usepackage{amssymb}
\usepackage{amsthm}
\usepackage{mathrsfs}
\usepackage{ascmac}
\usepackage[normalem]{ulem}
\usepackage{latexsym}
\usepackage{enumerate}
\usepackage{tikz}
\usetikzlibrary{positioning}
\usetikzlibrary{arrows.meta}
\tikzset{>={Latex[width=2mm,length=2mm]}}
\newtheorem{Thm}{Theorem}
\newtheorem{Lem}{Lemma}
\newtheorem{Cor}{Corollary}
\newcommand{\sX}{{\sf X}}
\newcommand{\sY}{{\sf Y}}
\newcommand{\sB}{{\sf B}}
\title{On Reachability Mixed Arborescence Packing}
\author{Tatsuya Matsuoka\thanks{The University of Tokyo, Japan (tatsuya\_matsuoka@mist.i.u-tokyo.ac.jp).}\and Shin-ichi Tanigawa\thanks{The University of Tokyo, Japan (tanigawa@mist.i.u-tokyo.ac.jp).}}
\date{August, 2018}
\begin{document}
\maketitle

\begin{abstract}
As a generalization of the Edmonds arborescence packing theorem, Kamiyama--Katoh--Takizawa~(2009) gave a good characterization of directed graphs that contain arc-disjoint arborescences spanning the set of vertices reachable from each root.
Fortier--Kir\'aly--L\'eonard--Szigeti--Talon~(2018) asked whether the result can be extended to mixed graphs by allowing both directed arcs and undirected edges.
In this paper, we solve this question by developing a polynomial-time algorithm for finding a collection of edge and arc-disjoint arborescences spanning the set of vertices reachable from each root in a given mixed graph.
\end{abstract}

\section{Introduction}
Edmonds' arborescence packing theorem~\cite{Edmonds1973} characterizes directed graphs that contain $k$ arc-disjoint spanning $r$-arborescences in terms of a cut condition.
Here, by an \emph{$r$-arborescence} we mean a subgraph of a given directed graph in which  there is exactly one path from $r$ to each vertex $v$, and a subgraph is said to be \emph{spanning} if its vertex set is equal to the whole vertex set.
The Edmonds theorem~\cite{Edmonds1973} is a directed counterpart of the celebrated Tutte~\cite{Tutte1961} and Nash-Williams~\cite{Nash-Williams1961} tree packing theorem.
Later Frank~\cite{Frank1978} gave a far reaching common generalization of these fundamental theorems by allowing both directed and undirected edges.
A \emph{mixed graph} $F=(V; E, A)$ is a graph consisting of the set $E$ of undirected edges and the set $A$ of directed arcs.
By regarding each undirected edge as a directed arc with both direction, each concept in directed graphs can be naturally extended.
For example, a \emph{mixed path} is a subgraph of a mixed graph in which each undirected edge can be orientated such that the resulting directed graph is a directed path, and a \emph{$r$-mixed arborescence} is a subgraph in which for each vertex $v$ there is exactly one mixed path from $r$ to $v$.

A packing of mixed arborescences is a collection of mutually edge and arc-disjoint mixed arborescences.
Frank~\cite{Frank1978} proved the following.
\begin{Thm}[Frank~\cite{Frank1978}]
\label{ThmMixed}
Let $F=(V; E, A)$ be a mixed graph, $r$ a node in $V$, and $k$ a positive integer.
Then there exists a packing of $k$ spanning $r$-mixed arborescences in $F$ if and only if
\[
e_E(\mathcal{P})+\sum_{i=1}^t\rho_A(V_i)\geq kt
\]
holds for every subpartition $\mathcal{P}=\{ V_1, \ldots , V_t\}$ of $V\setminus \{ r\}$, where $e_E$ denotes the number of edges in $E$ connecting distinct components in $\mathcal{P}$
and $\rho_A(V_i)$ denotes the number of arcs in $A$ entering $V_i$.
\end{Thm}
Frank's proof~\cite{Frank1978} is algorithmic, that is, he also gave  a polynomial-time algorithm for computing a packing in a given mixed graph.

In this paper we consider extending Frank's theorem~\cite{Frank1978} to packing {\em reachability arborescences}.
One serious limitation in the Edmonds theorem~\cite{Edmonds1973} is that arborescences to be packed are supposed to span the whole vertices.
Hence, if some node is not reachable from the root in an application, the theorem says nothing.
In such an instance, we are rather interesting in packing arborescences spanning the set of reachable vertices.
The following remarkable extension of the Edmonds theorem due to Kamiyama--Katoh--Takizawa~\cite{KKT2009} enables us to find such a packing even in the multi-root setting.
\begin{Thm}[Kamiyama--Katoh--Takizawa~\cite{KKT2009}]
\label{ThmReachability}
Let $D=(V, A)$ be a directed graph, and $r_1, \ldots , r_k\in V$.
Let $U_i\subseteq V$ $(i=1, \ldots , k)$ be the set of vertices reachable from $r_i$ in $D$.
Then, there exists a packing of $r_i$-arborescences $(i=1, \ldots , k)$ spanning $U_i$ in $D$ if and only if
\begin{equation}
\rho_A(X)\geq \left| \left\{ i\mid r_i\notin X,\ U_i\cap X\neq \emptyset \right\} \right| \label{EqKKT}
\end{equation}
holds for every $X\subseteq V$.
\end{Thm}
In this paper we unifies Theorem~\ref{ThmMixed} and Theorem~\ref{ThmReachability} by looking at reachability mixed arborescences in mixed graphs.
Specifically we consider the following problem.

\medskip
\noindent \textbf{The Reachability Mixed Arborescence Packing Problem}.
Given a mixed graph $F=(V; E, A)$ and $r_1,\dots, r_k\in V$, find a packing of $r_i$-mixed arborescences $T_i$ spanning $U_i$, where $U_i$ denotes the set of vertices reachable from $r_i$ by a mixed path in $F$. (See Figure~\ref{FigReachabilityMixed} for an example.)

\medskip
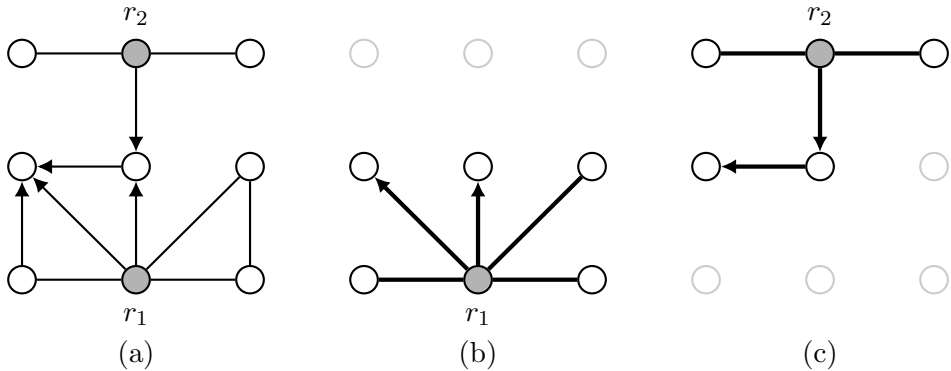
\begin{figure}[htb]
\centering
\begin{tikzpicture}
\tikzset{vertex/.style={fill = white, draw, circle, thick, minimum height=0.3cm}};
\tikzset{vertexN/.style={fill = white, draw = black!20, circle, thick, minimum height=0.3cm}};
\tikzset{root/.style={fill = black!30, draw, circle, thick, minimum height=0.3cm}};
\node[root] (r1a) at (-4.5,0) {};
\node[root] (r2a) at (-4.5,3) {};
\node[vertex] (v1a) at (-6,0) {};
\node[vertex] (v2a) at (-3,0) {};
\node[vertex] (v3a) at (-6,1.5) {};
\node[vertex] (v4a) at (-4.5,1.5) {};
\node[vertex] (v5a) at (-3,1.5) {};
\node[vertex] (v6a) at (-6,3) {};
\node[vertex] (v7a) at (-3,3) {};
\node [below= 0.05 of r1a] {$r_1$};
\node (a) at (-4.5,-1) {(a)};
\node [above= 0.05 of r2a] {$r_2$};
\foreach \u / \v in{r1a/v1a, r1a/v2a, r1a/v5a, r2a/v6a, r2a/v7a, v2a/v5a}
\draw[thick] (\u) -- (\v);
\foreach \u / \v in {r1a/v3a, r1a/v4a, r2a/v4a, v4a/v3a, v1a/v3a}
\draw[thick, ->] (\u) -- (\v);
\node[root] (r1b) at (0,0) {};
\node[vertexN] (r2b) at (0,3) {};
\node[vertex] (v1b) at (-1.5,0) {};
\node[vertex] (v2b) at (1.5,0) {};
\node[vertex] (v3b) at (-1.5,1.5) {};
\node[vertex] (v4b) at (0,1.5) {};
\node[vertex] (v5b) at (1.5,1.5) {};
\node[vertexN] (v6b) at (-1.5,3) {};
\node[vertexN] (v7b) at (1.5,3) {};
\node [below= 0.05 of r1b] {$r_1$};
\node (b) at (0,-1) {(b)};
\foreach \u / \v in{r1b/v1b, r1b/v2b, r1b/v5b}
\draw[ultra thick] (\u) -- (\v);
\foreach \u / \v in {r1b/v3b, r1b/v4b}
\draw[ultra thick, ->] (\u) -- (\v);
\node[vertexN] (r1c) at (4.5,0) {};
\node[root] (r2c) at (4.5,3) {};
\node[vertexN] (v1c) at (3,0) {};
\node[vertexN] (v2c) at (6,0) {};
\node[vertex] (v3c) at (3,1.5) {};
\node[vertex] (v4c) at (4.5,1.5) {};
\node[vertexN] (v5c) at (6,1.5) {};
\node[vertex] (v6c) at (3,3) {};
\node[vertex] (v7c) at (6,3) {};
\node (c) at (4.5,-1) {(c)};
\node [above= 0.05 of r2c] {$r_2$};
\foreach \u / \v in{r2c/v6c, r2c/v7c}
\draw[ultra thick] (\u) -- (\v);
\foreach \u / \v in {r2c/v4c, v4c/v3c}
\draw[ultra thick, ->] (\u) -- (\v);
\end{tikzpicture}
\caption{
(a) A mixed graph.
(b) An $r_1$-arborescence $T_1$ that spans $U_1$.
(c) An $r_2$-arborescence $T_2$ that spans $U_2$.
($T_1$ and $T_2$ are edge and arc-disjoint.)
}
\label{FigReachabilityMixed}
\end{figure}
This problem was recently introduced by Fortier--Kir\'aly--L\'eonard--Szigeti--Talon in~\cite{FKLST2018}.
They pointed out that a natural unification of the conditions in Theorem~\ref{ThmMixed} and Theorem~\ref{ThmReachability} is not sufficient, and developing a polynomial-time algorithm remains unsolved.
In this paper we give a solution to their question and develop the first polynomial-time algorithm.
\section{Packing in mixed graphs and covering supermodular functions}
Let $V$ be a finite set.
A family $\mathcal{H}$ of subsets of $V$ is called \emph{intersecting} if $X\cup Y, X\cap Y\in \mathcal{H}$ for every $X, Y\in \mathcal{H}$ with $X\cap Y\neq \emptyset$.
A function $f\colon 2^V\to \mathbb{R}$ is called \emph{supermodular} if
\begin{equation}
f(X)+f(Y)\leq f(X\cup Y)+f(X\cap Y) \label{EqSupermodular}
\end{equation}
holds for every $X, Y\subseteq V$.
For an intersecting family $\mathcal{H}$, a set function $f\colon \mathcal{H}\to \mathbb{R}$ is called \emph{intersecting supermodular} if~\eqref{EqSupermodular} holds for every $X, Y\in \mathcal{H}$ such that $X\cap Y\neq \emptyset$.

Frank's theorem (Theorem~\ref{ThmMixed}) can be  understood within the theory of \emph{supermodular function covering}.
For a set function $f\colon 2^V\to \mathbb{Z}$ on a finite set $V$, a directed graph $D=(V,A)$ is said to \emph{cover} $f$ if $\rho_A(X)\geq f(X)$ holds for every $X\subseteq V$.
The Edmonds arborescence packing theorem~\cite{Edmonds1973} states that a directed graph $D$ has a packing of $k$ spanning  $r$-arborescences if and only if $D$ covers the function $g_1$ defined by $g_1(X)=k$ if $\emptyset \neq X  \subseteq V\setminus \{r\}$ and $g_1(X)=0$ otherwise.
Using this formulation, a packing of spanning $r$-mixed arborescnces in a mixed graph $F=(V; E, A)$ can be considered as finding an orientation of the undirected graph $(V, E)$ such that the resulting directed graph covers $g_1-\rho_A$.
Frank~\cite{Frank1978} gave the following necessary and sufficient condition for an undirected graph to have an orientation covering an intersecting supermodular function.
\begin{Thm}[Frank~\cite{Frank1978}]
Let $G=(V, E)$ be an undirected graph, $\mathcal{H}\subseteq 2^V$ an intersecting family with $\emptyset \notin \mathcal{H}$ and  $V\in \mathcal{H}$, and $f\colon \mathcal{H}\to \mathbb{R}$ an intersecting supermodular function with $f(V)=0$.
There exists an orientation of $E$ that covers $f$ if and only if
\[
e_E(\mathcal{P})\geq \sum_{i=1}^tf(V_i)
\]
holds for every collection $\{ V_1, \ldots , V_t\}$ of mutually disjoint members of $\mathcal{H}$.
\label{ThmOrientation}
\end{Thm}
Theorem~\ref{ThmMixed} follows as a corollary of Theorem~\ref{ThmOrientation} since $g_1-\rho_A$ is intersecting supermodular.
See e.g.,~\cite{FK2009} for a survey.

When applying the above framework to the reachability mixed arborescence packing problem there is a trouble: the set function defined by the right side of~\eqref{EqKKT} is not intersecting supermodular.
This difficulty can be overcome by introducing \emph{bi-sets} as first observed  by B\'erczi--Frank~\cite{BF2010}.
Recently  Kir\'aly--Szigeti--Tanigawa~\cite{KST2018} showed that the condition of Theorem~\ref{ThmReachability} is equivalent to covering an interesecting supermodular \emph{bi-set} function.
Thus the reachability mixed arborescence packing problem can be formulated as a problem of finding an orientation of an undirected graph that covers an intersecting supermodular \emph{bi-set} function.
The problem of finding an orientation covering an intersecting or crossing supermodular bi-set function had been recognized as a prominent open problem as it contains the (rooted) $k$-connected orientation problem as a special case, and in 2012 the problem was shown to be NP-hard by Durand de Gevigney~\cite{Durand de Gevigney2012}.
Although this negative result shows the difficulty for developing the counterpart theory for bi-set functions, still one may wonder if some special cases can be solved efficiently.
In the next section we show that the reachability mixed arborescences packing problem is a solvable example.
\section{Characterization and Algorithm}
In this section, we give a polynomial-time algorithm and a characterization for the reachability mixed arborescence packing problem.
In Section~\ref{SubsecReductionToSupmodBiSet}, we first reduce the reachability mixed arborescence packing problem to a problem of computing an orientation that covers an intersecting supermodular \emph{bi-set} function based on the result of~\cite{KST2018}.
In Section~\ref{SubsecReductionToSupmodSet}, we further reduce the problem to a problem of computing an orientation that covers a supermodular \emph{set} function.
Then the algorithm follows from Frank's result on supermodular covering orientation~\cite{Frank1978}.
In Section~\ref{SubsecCharacterization}, we give a characterization.
\subsection{Reduction to a problem of covering an intersecting supermodular bi-set function}
\label{SubsecReductionToSupmodBiSet}
Let $V$ be a finite set.
A bi-set $\sX=(\sX_O, \sX_I)$ is a pair of sets satisfying $\sX_I\subseteq \sX_O\subseteq V$.
Let $\mathcal{P}_2(V)$ denote the family of all bi-sets.
For a bi-set $\sX=(\sX_O, \sX_I)$, let
\[
\rho (\sX):=|\{ (u, v)\mid u\in V\setminus \sX_O, v\in \sX_I\} |.
\]
For $\sX=(\sX_O, \sX_I)$ and $\sY=(\sY_O, \sY_I)$, let
\begin{align*}
\sX \sqcup \sY&:=(\sX_O\cup \sY_O, \sX_I\cup \sY_I),\\
\sX \sqcap \sY&:=(\sX_O\cap \sY_O, \sX_I\cap \sY_I).
\end{align*}
A bi-set family $\mathcal{H}\subseteq \mathcal{P}_2(V)$ is called \emph{intersecting} if $\sX \sqcup \sY , \sX \sqcap \sY \in \mathcal{H}$ hold for every $\sX , \sY \in \mathcal{H}$ with $\sX_I\cap \sY_I \neq \emptyset$.
A bi-set function $f\colon \mathcal{H}\to \mathbb{R}$ is called \emph{intersecting supermodular} if
\[
f(\sX)+f(\sY)\leq f(\sX\sqcup \sY)+f(\sX\sqcap \sY)
\]
holds for every $\sX, \sY\in \mathcal{H}$ with $\sX_I\cap \sY_I\neq \emptyset$.

We introduce some key notation following B\'erczi--Kir\'aly--Kobayashi~\cite{BKK2016}.
For $u, v\in V$, let $u\sim v$ if $\{ i\mid u\in U_i\} =\{ i\mid v\in U_i\}$.
This relation $\sim$ becomes an equivalence relation.
Let equivalence classes for $\sim$ be denoted by $\Gamma_1, \ldots , \Gamma_l$, and we call each $\Gamma_j$ an \emph{atom}.
Let
\begin{align*}
\mathcal{F}_j&:=\{ \sX\in \mathcal{P}_2(V)\mid \emptyset \neq \sX_I\subseteq \Gamma_j, (\sX_O\setminus \sX_I)\cap \Gamma_j=\emptyset \} &(j=1, \ldots , l),\\
\mathcal{F}&:=\bigcup_{j=1}^l\mathcal{F}_j,\\
p(\sX)&:=|\{ r_i\mid \sX_I\subseteq U_i, r_i\notin \sX_I, (\sX_O\setminus \sX_I)\cap U_i=\emptyset \} |&(\forall \sX\in \mathcal{F}).
\end{align*}
It is known that $\mathcal{F}$ is an intersecting family and $p$ is an intersecting supermodular bi-set function~\cite{KST2018}.
Moreover, the following facts are known.
\begin{Lem}[Lemma~5.3 of \cite{KST2018}, adapted]
Let $D=(V, A)$ be a directed graph, and $r_1, \ldots , r_k\in V$.
For any $\tilde{A}\subseteq A$,
\[
\rho_{\tilde{A}}(X)\geq |\{ i\mid r_i\notin X, U_i\cap X\neq \emptyset \} |
\]
holds for every $X\subseteq V$ if and only if
\[
\rho_{\tilde{A}}(\sX)\geq p(\sX)
\]
holds for every $\sX\in \mathcal{F}$.
\label{LemBiSets}
\end{Lem}
\begin{Lem}[Lemma~5.5 of \cite{KST2018}, adapted]
Let $D=(V, A)$ be a directed graph, $r_1, \ldots , r_k\in V, \tilde{A}\subseteq A$, and $\tilde{D}=(V, \tilde{A})\subseteq D$.
Let $\tilde{U}_i$ denote the set of vertices reachable from $r_i$ in $\tilde{D}$.
Then, if
\[
\rho_{\tilde{A}}(\sX)\geq p(\sX)
\]
holds for every $\sX\in \mathcal{F}$, then
\[
\{ i\mid v\in \tilde{U}_i\} =\{ i\mid v\in U_i\}
\]
holds for every $v\in V$.
\label{LemRoots}
\end{Lem}
\begin{Cor}
Let $F=(V; E, A)$ be a mixed graph, and $r_1, \ldots , r_k\in V$.
There exists a feasible solution if and only if there exists an orientation of $E$ such that in the resulting directed graph $\vec{F}=(V, \vec{A})$,
\[
\rho_{\vec{A}}(\sX)\geq p(\sX)
\]
holds for every $\sX\in \mathcal{F}$.
\label{CorReachabilityMixed}
\end{Cor}
\begin{proof}
First, we show necessity.
Suppose there exists a reachability mixed arborescence packing.
Then, each mixed arborescence with a fixed root $r_i$ can be uniquely converted to an ordinary arborescence with root $r_i$ by orientating undirected edges in $E$.
Let the resulting graph of this orientation be denoted by $\vec{F}=(V, \vec{A})$.
Then by Theorem~\ref{ThmReachability},
\[
\rho_{\vec{A}}(X)\geq |\{ i\mid r_i\notin X, U_i\cap X\neq \emptyset \} |
\]
holds for every $X\subseteq V$.
Hence by Lemma~\ref{LemBiSets},
\[
\rho_{\vec{A}}(\sX)\geq p(\sX)
\]
holds for every $\sX\in \mathcal{F}$.

Next, we show sufficiency.
Let $\vec{F}=(V, \vec{A})$ be a mixed graph obtained by orientating $F$ such that
\[
\rho_{\vec{A}}(\sX)\geq p(\sX)
\]
holds for every $\sX\in \mathcal{F}$.
Let $\tilde{U}_i$ denote the set of vertices reachable from $r_i$ in $\vec{F}$.
Then by Lemma~\ref{LemRoots},
\begin{equation}
\{ i\mid v\in \tilde{U}_i\} =\{ i\mid v\in U_i\} \label{EqSameReachability}
\end{equation}
holds for every $v\in V$.
By~\eqref{EqSameReachability} and Lemma~\ref{LemBiSets},
\begin{align*}
\rho_{\vec{A}}(X)&\geq |\{ i\mid r_i\notin X, \tilde{U}_i\cap X\neq \emptyset \} |\\
&=|\{ i\mid r_i\notin X, U_i\cap X\neq \emptyset \} |
\end{align*}
holds for every $X\subseteq V$.
Therefore, by Theorem~\ref{ThmReachability}, there is a packing of $r_i$-reachability arborescences in $\vec{F}$.
The corresponding packing in the original mixed graph $F$ is a reachability mixed arborescence packing by Lemma~\ref{LemRoots}.
\end{proof}
We also remark the following easy observation.
\begin{Lem}
For $\sX, \sY\in \mathcal{F}$ with $\sX_I=\sY_I$ and $\sX_O\subseteq \sY_O$, $p(\sX)\geq p(\sY)$ holds.
\label{LemBiSetOuterSet}
\end{Lem}
\begin{proof}
\begin{align*}
p(\sX)&=|\{ r_i\mid \sX_I\subseteq U_i, r_i\notin \sX_I, (\sX_O\setminus \sX_I)\cap U_i=\emptyset \} |\\
&\geq |\{ r_i\mid \sY_I\subseteq U_i, r_i\notin \sY_I, (\sY_O\setminus \sY_I)\cap U_i=\emptyset \} |&(\mbox{by }\sX_I=\sY_I\mbox{ and }\sX_O\subseteq \sY_O)\\
&=p(\sY).
\end{align*}
\end{proof}
\subsection{Reduction to a problem of covering an intersecting supermodular function}
\label{SubsecReductionToSupmodSet}
Let $F=(V; E, A)$ be a mixed graph.
For $X\subseteq V$, let $E[X]$ (resp. $A[X]$) be the set of edges (resp. arcs) whose both endpoints are in $X$.
Also, let $\partial_A^-(X)$ be the set of arcs in $A$ entering $X$.
For a directed graph $D=(V, A)$, let $\rho_D=\rho_A$.

Let $\Gamma_1, \ldots , \Gamma_l$ be the set of atoms in $F$ with respect to given roots $r_1, \ldots , r_k$.
For each atom $\Gamma_j$, we prepare a new mixed graph $F_j:=(V_j; E_j, A_j)$ obtained from $(\Gamma_j; E[\Gamma_j ], A[\Gamma_j ])$ by adding a new vertex $t_a$ for each $a\in \partial_A^-(\Gamma_j )$ and a new directed arc $t_av$ for each $a=uv\in \partial_A^-(\Gamma_j )$.
In other words,
\begin{align*}
V_j&:=\Gamma_j\cup \{ t_a\mid a\in \partial_A^-(\Gamma_j)\} ,\\
E_j&:=E[\Gamma_j] ,\\
A_j&:=A[\Gamma_j]\cup \{ t_av\mid a=uv\in \partial_A^-(\Gamma_j)\}
\end{align*}
for $j=1, \ldots , l$.
An important observation is that there exists no undirected edge connecting distinct atoms, and hence $\{ E_1, \ldots , E_l\}$ is a partition of $E$.

We say that $X\subseteq V_j$ is \emph{consistent} if for each $t_a\in X$, the head of $a\in A$ is in $X$.
For each $j$, we introduce a set family $\mathcal{H}_j\subseteq 2^{V_j}$, a bi-set $\sB(X)\in \mathcal{P}_2(V)$ for each $X\in \mathcal{H}_j$, and a set function $p_j\colon \mathcal{H}_j\to \mathbb{Z}$ as follows:
\begin{align}
\mathcal{H}_j&:=\{ X\subseteq V_j\mid X\cap \Gamma_j\neq \emptyset , X\mbox{ is consistent}\} , \notag \\
\sB(X)&:=((X\cap \Gamma_j)\cup \{ u\mid t_{uv}\in X\} , X\cap \Gamma_j)&(X\in \mathcal{H}_j), \notag \\
p_j(X)&:=p(\sB(X))&(X\in \mathcal{H}_j). \label{EqpjX}
\end{align}
For $\mathcal{H}_j$ and $p_j$, the following lemma holds.
\begin{Lem}
$\mathcal{H}_j$ is an intersecting family and $p_j$ is an intersecting supermodular function.
\label{LemHjpj}
\end{Lem}
\begin{proof}
Fix arbitrary $X, Y\in \mathcal{H}$ with $X\cap Y\neq \emptyset$.
For each $t_a\in X\cap Y$, since $X$ and $Y$ are consistent, the head of $a$ is in $X\cap Y$.
Hence $X\cap Y$ is consistent.
This also means that $X\cap Y\neq \emptyset$ implies $X\cap Y\cap \Gamma_j\neq \emptyset$.
Thus $X\cap Y\in \mathcal{H}_j$.
We can prove $X\cup Y\in \mathcal{H}_j$ by the same argument.

Next, we show that $p_j$ is intersecting supermodular.
For $\sB(X)$ and $\sB(Y)$ with $X, Y\in \mathcal{H}_j$, we have
\begin{align*}
\sB(X)\sqcap \sB(Y)&=((X\cap Y\cap \Gamma_j)\cup (\{ u\mid t_{uv}\in X\} \cap \{ u\mid t_{uv}\in Y\} ), X\cap Y\cap \Gamma_j)\mbox{ and}\\
\sB(X\cap Y)&=((X\cap Y\cap \Gamma_j)\cup \{ u\mid t_{uv}\in X\cap Y\} , X\cap Y\cap \Gamma_j)
\end{align*}
by definition.
Hence, by Lemma~\ref{LemBiSetOuterSet},
\begin{equation}
p(\sB(X\cap Y))\geq p(\sB(X)\sqcap \sB(Y)) \label{EqBiSetIntersection}
\end{equation}
holds.
For $X\cup Y$, one can check $\sB(X\cup Y)=\sB(X)\sqcup \sB(Y)$, implying
\begin{equation}
p(\sB(X\cup Y))=p(\sB(X)\sqcup \sB(Y)). \label{EqBiSetUnion}
\end{equation}
Also, for $X, Y\in \mathcal{H}_j$ with $X\cap Y\neq \emptyset$, $X\cap Y\cap \Gamma_j\neq \emptyset$ implies $(\sB(X))_I\cap (\sB(Y))_I\neq \emptyset$.
Hence by the intersecting supermodularity, we have
\begin{equation}
p(\sB(X))+p(\sB(Y))\leq p(\sB(X)\sqcup \sB(Y))+p(\sB(X)\sqcap \sB(Y)). \label{EqIntSupB}
\end{equation}
Putting all together, we obtain
\begin{align*}
&p_j(X)+p_j(Y)\\
={}&p(\sB(X))+p(\sB(Y))&(\mbox{by }\eqref{EqpjX})\\
\leq {}&p(\sB(X)\sqcup \sB(Y))+p(\sB(X)\sqcap \sB(Y))&(\mbox{by }\eqref{EqIntSupB})\\
\leq {}&p(\sB(X\cup Y))+p(\sB(X\cap Y))&(\mbox{by }\eqref{EqBiSetIntersection}\mbox{ and }\eqref{EqBiSetUnion})\\
={}&p_j(X\cup Y)+p_j(X\cap Y)&(\mbox{by }\eqref{EqpjX}).
\end{align*}
\end{proof}

Since $\{ E_1, \ldots , E_l\}$ is a partition of $E$, an orientation of $E$ gives not only an orientation of $F$ but also an orientation of $F_j$.
\begin{Lem}
Let $F=(V; E, A)$ be a mixed graph.
Let $\vec{F}$ and $\vec{F}_j$ be directed graphs obtained from $F$ and $F_j$ by orientating $E$, respectively.
Then,
\begin{equation}
\rho_{\vec{F}}(\sX)\geq p(\sX) \label{EqBiSetGeq}
\end{equation}
holds for every $\sX\in \mathcal{F}$ if and only if
\begin{equation}
\rho_{\vec{F}_j}(X)\geq p_j(X) \label{EqSetGeq}
\end{equation}
holds for every $j\in \{ 1, \ldots , l\}$ and for every $X\in \mathcal{H}_j$.
\label{LemAtoms}
\end{Lem}
\begin{proof}
For necessity, suppose~\eqref{EqBiSetGeq} holds for every $\sX\in \mathcal{F}$.
Take an arbitrary $j$ and $X\in \mathcal{H}_j$.
Let $\sX=\sB(X)$.
We show
\begin{equation}
\rho_{\vec{F}_j}(X)\geq \rho_{\vec{F}}(\sX). \label{EqSetGeqVec}
\end{equation}
To see~\eqref{EqSetGeqVec}, suppose $a=uv$ enters $\sX$.
Then $u\notin \sX_O$ and $v\in \sX_I$.
If $u\in \Gamma_j$, then since $X\cap \Gamma_j\subseteq \sX_O$ and $u\notin \sX_O$, $u\in \Gamma_j\setminus X$.
Thus, by $v\in X\cap \Gamma_j\subseteq X$, $a=uv$ enters $X$ as well.
If $u\notin \Gamma_j$, then, $t_{uv}\notin X$ and hence $t_{uv}v$ enters $X$.
This shows~\eqref{EqSetGeqVec} since $t_{a_1}$ and $t_{a_2}$ are distinct if $a_1$ and $a_2$ are distinct arcs entering $\sX$.
Therefore, by~\eqref{EqSetGeqVec}, \eqref{EqBiSetGeq}, and \eqref{EqpjX}, we obtain
\[
\rho_{\vec{F}_j}(X)\geq \rho_{\vec{F}}(\sX)\geq p(\sX)=p_j(X).
\]

For sufficiency, suppose~\eqref{EqSetGeq} holds for every $X\in \mathcal{H}_j$ for every $j\in \{ 1, \ldots , l\}$.
Take an arbitrary $\sX\in \mathcal{F}$.
Let
\[
X=\sX_I\cup \{ t_{v_1v_2}\mid v_1\in \sX_O\setminus \sX_I, v_2\in \sX_I\} .
\]
We show
\begin{equation}
\rho_{\vec{F}}(\sX)\geq \rho_{\vec{F}_j}(X). \label{EqBiSetGeqVec}
\end{equation}
To see~\eqref{EqBiSetGeqVec}, suppose $a'=u'v'$ enters $X$.
Then since there exists no arc entering $t_a$ for any $a\in \partial_A^-(\Gamma_j)$, we have $v'\in \sX_I$.
If $u'\in \Gamma_j$, then $u'\notin \sX_O\setminus \sX_I$ and thus $u'\notin \sX_O$.
Hence $u'v'$ enters $\sX$ as well.
If $u'\notin \Gamma_j$, then $u'\in \{ t_a\mid a\in \partial_A^-(\Gamma_j)\} \setminus \{ t_{v_1v_2}\mid v_1\in \sX_O\setminus \sX_I, v_2\in \sX_I\}$.
Let $uv'\in A$ be an arc such that $u'=t_{uv'}$.
Then $u\notin \sX_O$ and hence $uv'$ enters $\sX$.
This shows~\eqref{EqBiSetGeqVec} since $u_1$ with $u_1'=t_{u_1v'}$ and $u_2$ with $u_2'=t_{u_2v}$ are distinct if $u_1'v'$ and $u_2'v'$ are distinct arcs entering $X$.

Since $(\sB(X))_I=\sX_I$ and $(\sB(X))_O\subseteq \sX_O$, Lemma~\ref{LemBiSetOuterSet} implies
\begin{equation}
p(\sB(X))\geq p(\sX) \label{EqpBXpX}.
\end{equation}
Therefore, by~\eqref{EqBiSetGeqVec}, \eqref{EqSetGeq}, \eqref{EqpjX} and~\eqref{EqpBXpX},
\[
\rho_{\vec{F}}(\sX)\geq \rho_{\vec{F}_j}(X)\geq p_j(X)=p(\sB(X))\geq p(\sX)
\]
holds.
\end{proof}
By Lemma~\ref{LemAtoms}, the problem of covering a bi-set function $p$ can be reduced to the problem of covering the set function $p_j$ in $F_j$ for each $j$.
Since $\{ E_1, \ldots , E_l\}$ is a partition of $E$, we can solve each problem independently.
Thus, the following algorithm correctly solves the reachability mixed arborescence packing problem.

\medskip
\noindent \textbf{Algorithm} \textsc{Reachability Mixed Arborescence Packing}

\begin{description}
\item[Step 1] Construct $F_j$ $(j=1, \ldots , l)$ from the input mixed graph $F$.
\item[Step 2] Compute an orientation of $E_j$ that covers $p_j$ for each $F_j$.
If there exists no such orientation for some $j$, then there exists no reachability mixed arborescence packing and algorithm quits.
\item[Step 3] Construct a directed graph $\vec{F}$ obtained from $F$ by orientating $E$ using the orientation computed in Step~2.
\item[Step 4] Solve the reachability arborescence packing problem for $\vec{F}$.
\end{description}
Frank~\cite{Frank1978} showed that problem of covering an intersecting supermodular function by orientating edges can be solved in polynomial time.
Hence, Step~2 can be done in polynomial time.
Step~4 can be done in polynomial time due to~\cite{KKT2009}.
Therefore, the algorithm correctly solves the reachability mixed arborescence packing problem in polynomial time.
\subsection{Characterization}
\label{SubsecCharacterization}
By the identical argument to the proof of Lemma~\ref{LemAtoms}, we can show the following lemma.
\begin{Lem}
Let $F=(V; E, A)$ be a mixed graph, and $r_1, \ldots , r_k\in V$.
The following two statements are equivalent.

(i)
\[
e_E(\mathcal{P})+\sum_{i=1}^t\rho(\sX^i)\geq \sum_{i=1}^tp(\sX^i) \label{EqReachabilityMixed}
\]
holds for every family of bi-sets $\{ \sX^1, \ldots , \sX^t\}$ such that $\mathcal{P}=\{ (\sX^1)_I, \ldots , (\sX^t)_I\}$ is a subpartition of an atom $\Gamma_j$ and that $((\sX^i)_O\setminus (\sX^i)_I)\cap \Gamma_j=\emptyset$ holds for $i=1, \ldots , t$.

(ii)
\begin{equation}
e_{E_j}(\mathcal{P}')+\sum_{i=1}^t\rho(V^i)\geq \sum_{i=1}^tp_j(V^i) \label{EqLemEquiv}
\end{equation}
holds for every $j\in \{ 1, \ldots , l\}$ and for every subpartition $\mathcal{P}'=\{ V^1, \ldots , V^t\}$ of $V_j$.
\label{LemPP'Equiv}
\end{Lem}

We give the following characterization for the existence of a reachability mixed arborescence packing:
\begin{Thm}
Let $F=(V; E, A)$ be a mixed graph, and $r_1, \ldots , r_k$.
Let $U_i\subseteq V\ (i=1, \ldots , k)$ be the set of vertices reachable from $r_i$ in $F$.
Then, there exists a packing of $r_i$-mixed arborescences $(i=1, \ldots , k)$ spanning $U_i$ in $F$ if and only if
\begin{equation}
e_E(\mathcal{P})+\sum_{i=1}^t\rho_A(\sX^i)\geq \sum_{i=1}^t|\{ r_i\mid (\sX^i)_I\subseteq U_i, r_i\notin (\sX^i)_I, ((\sX^i)_O\setminus (\sX^i)_I)\cap U_i=\emptyset \} | \label{EqReachabilityMixed}
\end{equation}
holds for every family of bi-sets $\{ \sX^1, \ldots , \sX^t\}$ such that $\mathcal{P}=\{ (\sX^1)_I, \ldots , (\sX^t)_I\}$ is a subpartition of an atom $\Gamma_j$ and that $((\sX^i)_O\setminus (\sX^i)_I)\cap \Gamma_j=\emptyset$ holds for $i=1, \ldots , t$.
\label{ThmReachabilityMixed}
\end{Thm}
\begin{proof}
By Corollary~\ref{CorReachabilityMixed}, there exists a feasible packing if and only if there exists an orientation of $E$ such that in the resulting directed graph $\vec{F}=(V, \vec{A})$,
\[
\rho_{\vec{F}}(\sX)\geq p(\sX)
\]
holds for every $\sX\in \mathcal{F}$.
Due to Lemma~\ref{LemAtoms}, the above condition is equivalent to the existence of an orientation such that
\[
\rho_{\vec{F}_j}(X)\geq p_j(X)
\]
holds for every $j\in \{ 1, \ldots , l\}$ and for every $X\in \mathcal{H}_j$, or equivalently,
\[
\rho_{\vec{E}_j}(X)\geq p_j(X)-\rho_{A_j}(X)
\]
holds for every $j\in \{ 1, \ldots , l\}$ and for every $X\in \mathcal{H}_j$.

By Lemma~\ref{LemHjpj}, $\mathcal{H}_j$ is an intersecting family such that $\emptyset \notin \mathcal{H}_j$, $V_j\in \mathcal{H}_j$ and $p_j-\rho_{A_j}$ is an intersecting supermodular function such that $p_j(V_j)-\rho_{A_j}(V_j)=0$.
Thus, by Theorem~\ref{ThmOrientation}, there exists an orientation of $E_j$ that covers $p_j-\rho_{A_j}$ if and only if
\[
e_{E_j}(\mathcal{P}')\geq \sum_{i=1}^t(p_j(V^i)-\rho_{A_j}(V^i))
\]
holds for every subpartition $\{ V^1, \ldots , V^t\}$ of $V_j$.
Finally, by Lemma~\ref{LemPP'Equiv} this is equivalent to that~\eqref{EqReachabilityMixed} holds for every family of bi-sets $\{ \sX^1, \ldots , \sX^t\}$ such that $\mathcal{P}=\{ (\sX^1)_I, \ldots , (\sX^t)_I\}$ is a subpartition of an atom $\Gamma_j$ and that $((\sX^i)_O\setminus (\sX^i)_I)\cap \Gamma_j=\emptyset$ holds for $i=1, \ldots , t$.
\end{proof}
\section{Concluding Remarks}
We say that a mixed subgraph \emph{covers} a set function $f$ if it can be converted to a directed graph covering $f$ by orientating the undirected edges.
Khanna--Naor--Shepherd~\cite{KNS2005} obtained the weighted version of Frank's result~\cite{Frank1978} by showing that the problem of computing a minimum weight mixed subgraph covering an (intersecting/crossing) supermodular function can be reduced to the submodular flow problem.
Hence by using the reduction in Section~\ref{SubsecReductionToSupmodSet}, one can also solve the minimum weighted version of the reachability mixed arborescence packing problem.

Our result heavily relies on recent results on the reachability arborescence packing problem by~\cite{KST2018}, and these results were established in a more general setting by allowing additional matroid constraints (after a preprocessing shown in~\cite[Section~5.1]{KST2018}).
The result in this paper can be established in this general setting since the counterpart of Lemma~\ref{LemBiSetOuterSet} remains true.
\section*{Acknowledgement}
The first author is supported by JSPS Research Fellowship for Young Scientists.
The research of the first author was supported by Grant-in-Aid for JSPS Research Fellow Grant Number 16J06879.
The research of the second author was supported by JST CREST Grant Number JPMJCR14D2.


\begin{thebibliography}{99}
\bibitem{BF2010}
K. B\'erczi and A. Frank: Packing Arborescences.
In: S. Iwata (ed.), \textit{RIMS Kokyuroku Bessatsu} B23: Combinatorial Optimization and Discrete Algorithms, 1--31, 2010.
\bibitem{BKK2016}
K. B\'erczi, T. Kir\'aly, and Y. Kobayashi: Covering intersecting bi-set families under matroid constraints.
\textit{SIAM Journal on Discrete Mathematics}, \textbf{30}:1758--1774, 2016.
\bibitem{Durand de Gevigney2012}
O. Durand de Gevigney: On Frank's conjecture on $k$-connected orientations.
\textit{arXiv}, arXiv:1212.4086, 2012.
\bibitem{Edmonds1973}
J. Edmonds: Edge-disjoint branchings.
In: R. Rustin (ed.): \textit{Combinatorial Algorithms}, Academic Press, 91--96, 1973.
\bibitem{FKLST2018}
Q. Fortier, Cs. Kir\'aly, M. L\'eonard, Z. Szigeti, and A. Talon: Old and new results on packing arborescences in directed hypergraphs.
\textit{Discrete Applied Mathematics}, \textbf{242}:26--33, 2018.
\bibitem{Frank1978}
A. Frank: On disjoint trees and arborescences.
\textit{Algebraic Methods in Graph Theory}, Colloquia Mathematica Societatis J\'anos Bolyai, \textbf{25}:159--169, 1978.
\bibitem{FK2009}
A. Frank and T. Kir\'aly: A survey on covering supermodular functions.
In: W. Cook, L. Lov\'asz, and J. Vygen (eds.), \textit{Research Trends in Combinatorial Optimization}, Springer, 87--126, 2009.
\bibitem{KKT2009}
N. Kamiyama, N. Katoh, and A. Takizawa: Arc-disjoint in-trees in directed graphs.
\textit{Combinatorica}, \textbf{29}:197--214, 2009.
\bibitem{KNS2005}
S. Khanna, J. Naor, and F. B. Shepherd: Directed network design with orientation constraints.
\textit{SIAM Journal on Discrete Mathematics}, \textbf{19}:245--257, 2005.
\bibitem{KST2018}
Cs. Kir\'aly, Z. Szigeti, and S. Tanigawa: Packing of arborescences with matroid constraints via matroid intersection.
\textit{EGRES Technical Reports}, TR-2018-08, Egerv\'ary Research Group, 2018.
\bibitem{Nash-Williams1961}
C. St. J. A. Nash-Williams: Edge-disjoint spanning trees of finite graphs.
\textit{Journal of the London Mathematical Society}, \textbf{36}:445--450, 1961.
\bibitem{Tutte1961}
W. T. Tutte: On the problem of decomposing a graph into $n$ connected factors.
\textit{The Journal of the London Mathematical Society} \textbf{36}:221--230, 1961.
\end{thebibliography}
\end{document}